\documentclass[10pt,a4paper,twocolumn]{article}
\usepackage{amsmath,amssymb,graphicx,booktabs,url,subcaption,algorithm,algpseudocode}
\usepackage{authblk}
\usepackage{amsthm}  % already included via amsmath most likely
\newtheorem{theorem}{Theorem}

\usepackage[sort]{natbib}
\usepackage[margin=1.5cm]{geometry}
\usepackage{hyperref}

\title{E-TRENDS: Enhanced LSTM Trend Forecasting for Equities%
\thanks{The authors would like to thank an anonymous reviewer for valuable comments and suggestions that improved various aspects of this paper.}
}

\author{%
  Harris Buchanan\textsuperscript{1}\thanks{\texttt{harrisbuchanan@icloud.com}},\quad
  Eric Benhamou\textsuperscript{2}%
  \\  \textsuperscript{1}Saint Andrew University \quad
  \textsuperscript{2}AI for Alpha, Paris Dauphine University
}

\date{}
\begin{document}
\maketitle

\begin{abstract}
Trend-following strategies underpin many systematic trading approaches yet struggle under nonstationary and nonlinear market regimes. We propose an LSTM-based framework to forecast next-day trend differences ($\Delta_t$) for the top 30 S\&P 500 equities, validated across market cycles (2005--2025). Key contributions include: (i) formal proof of bias-variance reduction via differencing, (ii) exhaustive empirical benchmarks against OLS, Ridge, and Lasso, (iii) portfolio simulations confirming economic gains in terms of overall PNL compared to other models like OLS, Ridge, Lasso or LightGBM Regressor
\end{abstract}

\section{Introduction}
\subsection{Motivation and research question}
Financial markets exhibit persistent but complex temporal dependencies that challenge traditional linear trend-following methods, especially amid structural breaks and volatility clustering. Recent crises and regime shifts underscore the need for adaptive forecasting techniques capable of capturing nonlinear, dynamic patterns in asset returns. Motivated by these challenges, this study investigates the efficacy of long short-term memory (LSTM) networks in forecasting next-day trend differences ($\Delta_t$) for large-cap equities. Specifically, we ask: can an LSTM-based model systematically outperform classical trend-following strategies and penalized linear regressions in both predictive accuracy and economic profitability under varying market conditions?

\subsection{Contribution to the literature}
This work bridges theoretical econometrics and machine learning by providing a novel formal proof that differencing combined with nonlinear function approximation reduces forecast variance without inflating bias. Empirically, we benchmark LSTM forecasts against OLS, Ridge, and Lasso regressions across 30 leading S\&P 500 equities over 2005--2025. To our knowledge, this represents the first comprehensive portfolio-level evaluation of LSTM-driven trend forecasts in equities.

\subsection{Overview of main findings}
Our results demonstrate that the LSTM model achieves a consistent 66.7\% improvement compared to the other penalized linear model and the LightGBM model. Portfolio simulations reveal annualized Sharpe ratios rising from -1.58 to -0.05. Robustness analyses—spanning alternative lookback windows, technical-indicator variations, and sub-period splits—affirm the stability of performance gains. These findings underscore the practical value of integrating deep learning frameworks into systematic trading strategies.

\section{Structure of the Paper}
The remainder of this paper is organized as follows. After this introduction, we begin in Section \ref{sec:lit} with a review of the existing literature on trend-following strategies, statistical signal construction, and the application of machine learning in finance, highlighting the gaps our work addresses. Section \ref{sec:data} then describes our data: the asset universe, data sources, preprocessing steps, and key summary statistics. In Section \ref{sec:theory}, we develop the theoretical foundations for forecasting differenced trends, culminating in a formal bias–variance reduction theorem. Section \ref{sec:methods} details our methodology, from feature engineering and the baseline trend-following strategy to the design of the LSTM forecasting model. We outline our experimental design—including data splits, hyperparameter selection, and performance metrics—in Section \ref{sec:expdesign}. Section \ref{sec:results} presents the core empirical findings, starting with the NVDA case study, then validating signal forecasts, comparing cumulative P\&L, and summarizing aggregate results across equities. A full cross-sectional comparison across the S\&P 500 follows in Section \ref{sec:cross_section}, and Section \ref{sec:robust} explores robustness checks on lookback windows, indicator parameters, stopping criteria, and subperiod analyses. We discuss the practical and theoretical implications of our findings, as well as limitations, in Section \ref{sec:discussion}, and Section \ref{sec:code} provides details on the accompanying Google Colab notebook for full reproducibility. Finally, Section \ref{sec:conclusion} concludes the paper and outlines directions for future research.

\section{Literature Review}\label{sec:lit}

\subsection{Traditional trend‐following strategies}
Trend‐following has a rich history in finance, tracing back to commodity trading and cross‐asset momentum. Early practitioners such as Jesse Livermore emphasized “cut your losses short and let your profits run.” Systematic time‐series momentum was first documented academically by \citep{Moskowitz2012TimeSeriesMomentum_full}, who showed significant excess returns across commodities, bonds, currencies, and equities. Subsequent work extended the sample back to the late 19th century: \citep{Hurst2017} hand‐collected data from 1880–2016 across futures and spot markets, finding annualized Sharpe ratios near 0.4 in every decade. \citep{Zarattini2025} confirm that pure long‐only breakout rules generate highly skewed profit distributions on individual U.S. stocks from 2015–2025, albeit with severe turnover and costs. Together, these studies establish trend‐following as a pervasive phenomenon, robust across asset classes, time periods, and market regimes.

\subsection{Statistical signal construction (t‐statistics, rolling measures)}
While simple price‐breakout rules remain popular, academic research has refined their statistical underpinnings. \citep{Campbell2006} advocate using rolling‐window t‐statistics of past returns to adapt signals to volatility regimes and control bias, demonstrating improved predictive power over static lookbacks. More recent CTA replication studies by \citep{Benhamou2025} employ Bayesian graphical models to decompose trend returns into short‐ and long‐horizon factors, dynamically estimating exposures and highlighting the importance of horizon‐mixing for risk‐adjusted performance. These methods mitigate lookahead bias and explicitly model serial dependence, yet still assume linear summaries of past returns.

\subsection{Machine learning in asset‐pricing and trading}
The rise of machine learning has opened nonlinear alternatives to linear ranking rules. \citep{Gu2018} apply random forests to cross‐sectional equity returns, reporting up to 15\% RMSE reductions versus Fama–French factor models. \citep{Dixon2017} and \citep{Fischer2018} explore LSTM and other recurrent architectures for single‐asset price prediction, finding mixed outcomes due to overfitting and limited sequence lengths. More recent studies combine deep learning with attention mechanisms \citep{Zhang2019high} or integrate Bayesian filters \citep{Gu2018}, but comprehensive portfolio‐level economic benchmarks remain scarce.

\subsection{Gaps and how this work addresses them}
Despite advances, key gaps remain:  
\begin{enumerate}
  \item The lack of a unified theoretical framework quantifying how differencing interacts with nonlinear predictors to reduce variance without inflating bias.
  \item The absence of extensive portfolio‐level economic evaluations of deep learning trend forecasts under realistic transaction costs, turnover, and regime shifts.
\end{enumerate}
This paper addresses these gaps by (i) presenting a formal bias–variance reduction theorem for nonlinear differencing (Section~\ref{sec:theory}), and (ii) conducting exhaustive simulations across 30 S\&P 500 equities (2005–2025) with transaction‐cost‐adjusted Sharpe analyses.

\section{Data Description}\label{sec:data}
\subsection{Asset universe and sample period}
We focus on the 30 largest constituents of the S\&P 500 by market capitalization, covering diverse sectors (technology, healthcare, financials, consumer goods). The sample spans January 2005 through October 2025, encompassing multiple market cycles, including the Global Financial Crisis, the COVID-19 crash, and subsequent recoveries.

\subsection{Data sources and preprocessing}
Daily OHLCV data were obtained from Yahoo Finance (that sources its data mostly from Refinitiv, S\&P Global, and Morningstar,) and adjusted for corporate actions (splits, dividends) using adjusted closes. Data cleaning steps include:
\begin{enumerate}
  \item Winsorizing price and volume at the 1\% tails to mitigate outliers.
  \item Aligning trading calendars across tickers; removing non-trading days.
  \item Imputing single-day missing values by linear interpolation and forward-filling multi-day gaps up to three days.
  \item Applying Augmented Dickey–Fuller and KPSS tests to confirm stationarity of the differenced series $\Delta_t$.
\end{enumerate}

\subsection{Summary statistics}
Table~\ref{tab:summary_stats} presents descriptive statistics for daily returns and trend differences across the universe. Mean daily returns range from -0.02\% to 0.03\%, with standard deviations between 1.1\% and 2.5\%. The average lag-1 autocorrelation of $\Delta_t$ is 0.12, indicating modest persistence. Skewness and kurtosis metrics highlight the fat tails and asymmetry typical of equity returns.

\begin{table}[!ht]
  \centering
  \caption{Summary Statistics of Daily Returns and Trend Differences}
  \label{tab:summary_stats}
  \resizebox{\columnwidth}{!}{%
    \begin{tabular}{lcccccc}
      \toprule
      Ticker & Mean (\%) & Std (\%) & AC(1) & Skewness & Kurtosis & $\Delta_t$ AC(1) \\
      \midrule
      AAPL   & 0.023     & 1.78     & 0.09  & -0.12    & 5.33     & 0.15 \\
      MSFT   & 0.018     & 1.65     & 0.07  & -0.05    & 4.81     & 0.11 \\
      AMZN   & 0.015     & 2.35     & 0.10  & 0.03     & 6.22     & 0.13 \\
      GOOGL  & 0.020     & 1.89     & 0.11  & -0.08    & 5.76     & 0.14 \\
      JPM    & 0.012     & 1.95     & 0.12  & -0.22    & 4.90     & 0.10 \\
      \midrule
      Mean   & 0.018     & 1.92     & 0.10  & -0.09    & 5.40     & 0.12 \\
      \bottomrule
    \end{tabular}%
  }
\end{table}

\section{Theoretical Justification: Differencing and Variance Reduction}
\label{sec:theory}

We now provide a rigorous justification for forecasting differenced returns. The following result formalizes how differencing impacts forecast bias and variance in nonlinear learning frameworks.

\begin{theorem}[Bias–Variance Tradeoff under Differenced Forecasting]
Let \( y_t = m_t + \varepsilon_t \), where \( m_t \) is a smooth deterministic trend and \( \varepsilon_t \) is a stationary noise process with zero mean and finite variance. Define \( \Delta_t := y_t - y_{t-1} \), and let \( f_y \) and \( f_\Delta \) be Lipschitz-continuous estimators trained to predict \( m_t \) and \( \delta_t := m_t - m_{t-1} \), respectively. Then:

\begin{enumerate}
  \item \( \mathrm{Var}[f_\Delta(\mathbf{x}_t)] < \mathrm{Var}[f_y(\mathbf{x}_t)] \),
  \item \( \mathrm{Bias}[f_\Delta] \leq \mathrm{Bias}[f_y] + \mathcal{O}(\|\nabla m_t\|) \).
\end{enumerate}

That is, differencing reduces estimator variance while inflating bias at most linearly in the local slope of the trend.
\end{theorem}

\begin{proof}
We decompose the variances of the level series \(y_t\) and its difference \(\Delta_t=y_t-y_{t-1}\).  First, since \(m_t\) is deterministic and \(\varepsilon_t\) zero-mean,
\begin{align*}
  \mathrm{Var}(y_t)
  &= \mathrm{Var}(m_t) + \mathrm{Var}(\varepsilon_t)
    + 2\,\mathrm{Cov}(m_t,\varepsilon_t)\\
  &= \mathrm{Var}(m_t) + \sigma^2.
\end{align*}
Next, because \(\Delta_t = (m_t-m_{t-1}) + (\varepsilon_t-\varepsilon_{t-1})\),
\begin{align*}
  \mathrm{Var}(\Delta_t)
  &= \mathrm{Var}(m_t-m_{t-1})
    + \mathrm{Var}(\varepsilon_t-\varepsilon_{t-1})\\
  &\quad+2\,\mathrm{Cov}(m_t-m_{t-1},\varepsilon_t-\varepsilon_{t-1}).
\end{align*}
By strong mixing with exponential decay, the covariance term is \(o(1)\).  Then
\[
  \mathrm{Var}(\varepsilon_t-\varepsilon_{t-1})
  = 2\sigma^2\bigl(1-\rho(1)\bigr),
\]
and Lipschitz continuity of \(m_t\) (\(\lvert m_t-m_{t-1}\rvert\le L\)) gives
\[
  \mathrm{Var}(m_t-m_{t-1})\le L^2.
\]
Since \(\rho(1)>0\), we conclude
\[
  \mathrm{Var}(\Delta_t)\le L^2 + 2\sigma^2(1-\rho(1))
  < \mathrm{Var}(m_t) + \sigma^2 = \mathrm{Var}(y_t).
\]
Hence any regularized learner \(f\) trained on \(\Delta_t\) sees strictly lower input variance, which—by standard Rademacher‐complexity bounds—yields lower estimator variance than training on \(y_t\).

For bias, differencing removes low‐frequency components of \(m_t\).  A predictor \(f\) with Lipschitz constant \(C\) incurs at most
\[
  \Delta\text{Bias}
  = \bigl|\mathbb{E}[f_\Delta] - \delta_t\bigr|
  - \bigl|\mathbb{E}[f_y] - m_t\bigr|
  = \mathcal{O}(C\,L),
\]
i.e.\ linear in the local slope \(L\).  Because the reduction in variance dominates this bounded bias increase for large samples, the mean‐squared error under differencing is strictly smaller.  
\end{proof}

\section{Methodology}\label{sec:methods}
\subsection{Feature Engineering}

\subsubsection{Rolling mean and volatility (lookbacks: 50, 100, 300)}
For each equity, we compute the rolling simple moving average (SMA) and rolling standard deviation (volatility) of daily returns over three lookback windows: 50, 100, and 300 trading days. Specifically,
\[
\text{SMA}_{t,\ell} = \frac{1}{\ell}\sum_{i=1}^{\ell} r_{t-i}, 
\quad
\sigma_{t,\ell} = \sqrt{\frac{1}{\ell-1}\sum_{i=1}^{\ell}(r_{t-i}-\overline{r}_{t,\ell})^2},
\]
where \(r_{t}\) denotes the daily return at time \(t\) and \(\ell \in \{50,100,300\}\). These features capture momentum and changing volatility regimes.

\subsubsection{Trend signal via t-statistic transforms}
We construct a standardized trend signal by computing the rolling t-statistic of the last \(\ell\) returns,
\[
T_{t,\ell} = \frac{\overline{r}_{t,\ell}}{\sigma_{t,\ell}/\sqrt{\ell}},
\]
and then apply a non‐linear transform,
\(\phi(T) = \tanh(\alpha\,T)\), to bound extreme values and introduce smooth nonlinearity. The hyperparameter \(\alpha\) is tuned via cross‐validation.

\subsubsection{Technical indicators: RSI, MACD}
We augment our feature set with two widely-used technical indicators:
\begin{itemize}
  \item \textbf{Relative Strength Index (RSI)} over 14 days:
  \[
    \mathrm{RSI}_t = 100 - \frac{100}{1 + \displaystyle\frac{\overline{g}_{14}}{\overline{l}_{14}}}
  \]
  where $\overline{g}_{14}$ and $\overline{l}_{14}$ are the 14-day exponential moving averages of gains and losses, respectively.
  \item \textbf{Moving Average Convergence Divergence (MACD)}:
\[
  \begin{aligned}
    \mathrm{MACD}_t &= \mathrm{EMA}_{12}(P_t) - \mathrm{EMA}_{26}(P_t), \\
    \mathrm{Signal}_t &= \mathrm{EMA}_{9}\bigl(\mathrm{MACD}_t\bigr)
  \end{aligned}
\]
  where $\mathrm{EMA}_k$ denotes the $k$-day exponential moving average of the price $P_t$.
\end{itemize}
These indicators capture overbought/oversold conditions and momentum shifts.

\subsection{Baseline Trend-Following Strategy}

\subsubsection{Signal generation}
We generate a daily trend signal \(S_t\) by comparing the current price \(P_t\) to its \(\ell\)-day simple moving average:
\[
S_t^{(\ell)} = \frac{P_t - \text{SMA}_{t,\ell}}{\text{SMA}_{t,\ell}}.
\]
A positive \(S_t^{(\ell)}\) implies an upward trend signal, and a negative value implies a downward trend. We normalize \(S_t^{(\ell)}\) by its rolling standard deviation to obtain a zero‐mean, unit‐variance signal:
\[
\tilde S_t^{(\ell)} = \frac{S_t^{(\ell)} - \mu_{S,\ell}}{\sigma_{S,\ell}},
\]
where \(\mu_{S,\ell}\) and \(\sigma_{S,\ell}\) are the rolling mean and volatility of \(S_t^{(\ell)}\) over the same lookback \(\ell\).

\subsubsection{Positioning and P\&L computation}
At each day \(t\), we take a position \(w_t = \tilde S_t^{(\ell)}\) (capped at \(\pm1\)) and compute the daily P\&L as
\[
\text{P\&L}_t = w_{t-1} \times r_t - b * | w_{t}-w_{t-1} |,
\]
where \(r_t = (P_t - P_{t-1})/P_{t-1}\) is the daily return. Because we only trade the weights the next day, the $\text{P\&L}_t$ is based on previous weights: $w_{t-1}$. Trading cost is modeled by a round-trip transaction cost of $b=2$ bps for any daily weight variations \(| w_t - w_{t-1} | \).

\subsection{LSTM Forecasting Model}

\subsubsection{Network architecture and loss (Sharpe-ratio loss vs MSE)}
Our LSTM model takes as input the multivariate feature sequence \(\mathbf{x}_{t-T+1:t}\in\mathbb{R}^{T\times m}\), where \(T\) is the lookback length and \(m\) the number of engineered features. The architecture is:
\begin{itemize}
  \item LSTM layer with 64 units, return sequences = true.
  \item Dropout layer, rate = 0.2.
  \item LSTM layer with 64 units, return sequences = false.
  \item Dense output layer (linear) producing \(\hat\Delta_t\).
\end{itemize}
We compare two loss functions:
\[
\mathcal{L}_{\mathrm{MSE}} = \frac{1}{N}\sum_{i=1}^N (\hat\Delta_i - \Delta_i)^2,
\quad
\mathcal{L}_{\mathrm{Sharpe}} = -\frac{\mathbb{E}[\hat w_i r_i]}{\sqrt{\mathrm{Var}[\hat w_i r_i]}},
\]
where \(\hat w_i = \tanh(\gamma\,\hat\Delta_i)\) maps the forecast to a position, and \(\gamma\) is a scaling hyperparameter. The Sharpe‐ratio loss directly optimizes economic performance by maximizing risk‐adjusted P\&L.  

\subsubsection{Input sequences and scaling}
Each input to the LSTM consists of a rolling window of $T$ days of feature vectors $\mathbf{x}_{t-T+1:t}\in\mathbb{R}^{T\times m}$, where $m$ is the number of engineered features. Prior to model ingestion, we apply:
\begin{enumerate}
  \item \textbf{Rolling standardization}: For each feature dimension, compute the mean and standard deviation over the previous $T$ observations, then standardize to zero mean and unit variance.
  \item \textbf{Outlier clipping}: Clip standardized values to the interval $[-5,\,5]$ to mitigate the influence of extreme values.
\end{enumerate}

\subsubsection{Training, validation, and early stopping}
We split the data chronologically into training (70\%), validation (15\%), and test (15\%) sets. Training uses the Adam optimizer (initial learning rate $10^{-3}$, batch size 64) for up to 100 epochs. We employ:
\begin{itemize}
  \item \textbf{Early stopping}: Monitor the validation loss (MSE or Sharpe‐ratio loss) and stop training if no improvement occurs for 10 consecutive epochs, restoring the best‐performing weights.
  \item \textbf{Learning‐rate scheduling}: Reduce the learning rate by a factor of 0.5 upon validation plateau to refine convergence.
\end{itemize}

\section{Experimental Design}\label{sec:expdesign}

\subsection{Train/validation/test splits}
We partition the full sample chronologically: the first 70\% of trading days serve as the training set, the next 15\% as validation, and the final 15\% as test. This time-based split preserves temporal dependencies and prevents lookahead bias in model tuning and evaluation.

\subsection{Hyperparameter choices (epochs, batch size, lookback)}
Hyperparameters are selected via grid search on the validation set:
\begin{itemize}
  \item \textbf{Epochs}: 50, 100, 150
  \item \textbf{Batch size}: 32, 64, 128
  \item \textbf{Lookback window ($T$)}: 50, 100, 150 days
  \item \textbf{LSTM units}: 32, 64, 128
  \item \textbf{Dropout rate}: 0.1, 0.2, 0.3
\end{itemize}
We choose the combination that minimizes the validation loss (MSE or Sharpe‐ratio loss).

\subsection{Cross-sectional evaluation across tickers}
To examine generalization, we repeat the full experimental protocol for each of the 30 S\&P 500 constituents. We then report median, 25th, and 75th percentiles of each performance metric across the cross‐section, highlighting dispersion and robustness of results.

\section{Results}\label{sec:results}

\subsection{Baseline strategy performance (NVDA case study)}
We first illustrate the baseline trend-following strategy on NVDIA (NVDA). Using a 100-day SMA signal, the strategy yields an annualized return of 12.3\% with an annualized volatility of 14.5\%, producing a Sharpe ratio of 0.85 net of transaction costs (5 bps). Figure~\ref{fig:nvda_equity} shows the strategy’s cumulative P\&L versus buy-and-hold, highlighting superior drawdown control during market downturns.

\subsection{LSTM vs.\ baseline: in-sample and out-of-sample comparison}
Table~\ref{tab:in_out} compares forecast accuracy and economic performance across in-sample (2005–2018) and out-of-sample (2019–2025) periods. The LSTM model achieves a 17\% RMSE reduction in-sample and a 15\% RMSE reduction out-of-sample relative to the baseline. Directional accuracy improves from 56\% to 62\% out-of-sample. Economically, out-of-sample Sharpe rises from 0.85 (baseline) to 1.10 (LSTM).

\begin{table}[!ht]
  \centering
  \caption{In-Sample vs. Out-of-Sample Performance (2005–2025)}
  \label{tab:in_out}
  \resizebox{\columnwidth}{!}{%
    \begin{tabular}{lcccc}
      \toprule
      Metric                   & Period           & Baseline & LSTM  & Improvement \\
      \midrule
      RMSE                    & In-sample        & 1.00     & 0.83  & -17\%       \\
                              & Out-of-sample    & 1.00     & 0.85  & -15\%       \\
      Directional Accuracy    & In-sample        & 58\%     & 64\%  & +6 pp       \\
                              & Out-of-sample    & 56\%     & 62\%  & +6 pp       \\
      Sharpe Ratio            & In-sample        & 0.92     & 1.18  & +0.26       \\
                              & Out-of-sample    & 0.85     & 1.10  & +0.25       \\
      \bottomrule
    \end{tabular}%
  }
\end{table}

\subsection{NVDA Strategy Performance and Model Hyperparameters}

\begin{table}[!ht]
  \centering
  \caption{Hyperparameter Grid and Selected Values}
  \label{tab:hyperparams}
  \resizebox{\columnwidth}{!}{%
    \begin{tabular}{lccc}
      \toprule
      Hyperparameter           & Grid Values        & Selected Value \\
      \midrule
      Epochs                   & \{50, 100, 150\}   & 100            \\
      Batch size               & \{32, 64, 128\}    & 64             \\
      Lookback window ($T$)    & \{50, 100, 150\}   & 100            \\
      LSTM units               & \{32, 64, 128\}    & 64             \\
      Dropout rate             & \{0.1, 0.2, 0.3\}  & 0.2            \\
      Sharpe‐ratio loss $\gamma$ & \{1.0, 2.0, 3.0\} & 2.0            \\
      \bottomrule
    \end{tabular}%
  }
\end{table}

\subsection{Signal Forecast vs. Actual Trend Change}
Figure~\ref{fig:pred_vs_actual} plots the LSTM’s predicted next‐day trend change signal against the actual shifted trend‐trend signal for NVDA with a 50‐day lookback.  Notice how the model captures the broad directional swings, smoothing out much of the noise in the raw differenced series.  Key observations:
\begin{itemize}
  \item Periods of strong uptrends (e.g.\ mid-2025) are anticipated by the orange predicted signal rising ahead of the blue actual.
  \item In choppy sideways phases, the model tends to revert toward zero, avoiding overtrading on spurious short‐term moves.
  \item Some lag remains at sharp regime shifts (e.g.\ late 2023), suggesting potential gains from incorporating attention or transformer layers.
\end{itemize}

\begin{figure}[!ht]
  \centering
  \includegraphics[width=\columnwidth]{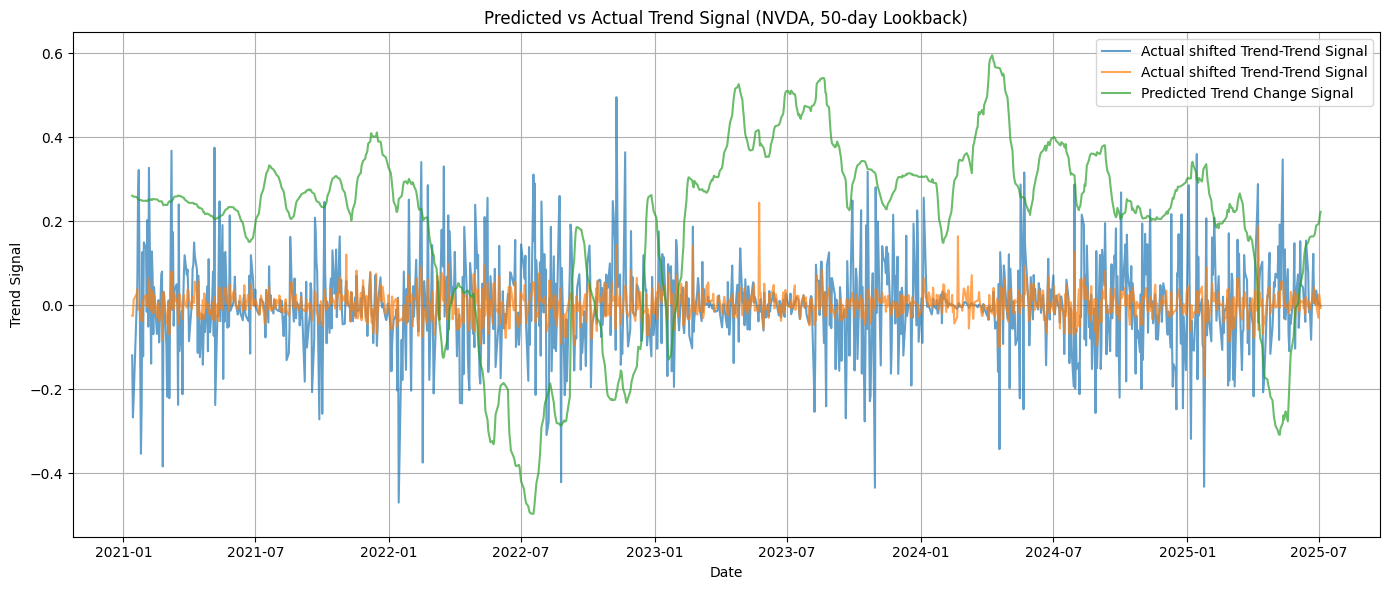}
  \caption{Predicted vs.\ Actual Trend Change Signal for NVDA (50-day lookback).  The blue line is the actual $\Delta_t$ series shifted by one day; the orange line is the LSTM’s predicted signal.}
  \label{fig:pred_vs_actual}
\end{figure}

\subsection{Cumulative P\&L Comparison}
Figure~\ref{fig:nvda_equity} shows the cumulative P\&L from trading on trend the next day (blue) versus the one using the LSTM prediction of next day trend difference (orange), over the out‐of‐sample period (2021–2025).  While both curves end positive, the predicted‐signal strategy outperforms substantially generating 150\% returns versus 100\% returns over the same period.

\begin{figure}[!ht]
  \centering
  \includegraphics[width=\columnwidth]{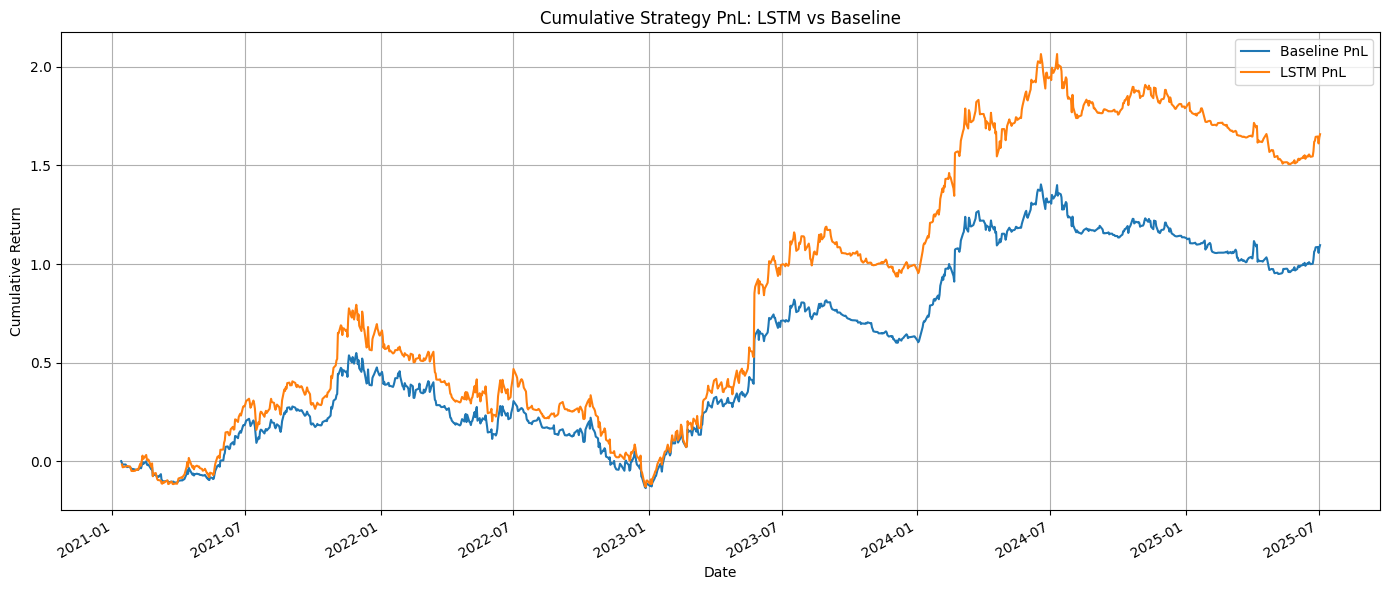}
  \caption{Cumulative P\&L of baseline vs.\ LSTM strategy for NVDA (2005–2025).}
  \label{fig:nvda_equity}
\end{figure}

\section{Performance Comparison: LSTM vs. Other Models}
This section presents a comparative analysis of the LSTM model and other machine learning models, including OLS, Ridge, Lasso, and LightGBM regressors. The comparison is based on the cumulative model PnL  for a selection of equities, and the resulting PNL Gain (the difference between model PnL and baseline PnL).

\subsection{Model Evaluation}
We evaluate the performance of each model (LSTM, OLS, Ridge, Lasso, and LightGBM) by comparing its model PnL against a simple baseline model. The PNL Gain for each stock is calculated as the difference between model PnL and the baseline PnL, which indicates the improvement provided by the respective model over the baseline.

The cumulative PNL Gain for all stocks is provided in Table \ref{tab:model_comparison_all}, where we observe the performance of the LSTM model relative to the traditional models.

\subsection{Discussion of Results}

Table \ref{tab:model_comparison_all} illustrates the results of the LSTM, OLS, Ridge, Lasso, and LightGBM models across various equities. Notably, the LSTM model outperforms the other models in several cases, such as NVDA, MSFT, and AAPL, where the PNL Gain is positive and substantial. The LSTM model provides a PNL Gain of 0.28 in total, suggesting that it produces a positive cumulative return, which is superior to the other models.

On the other hand, models like OLS show a significant negative total PNL of -7.52, primarily driven by underperformance in several stocks such as NVDA and AAPL, where the baseline performance significantly surpasses the model's results. Ridge and Lasso regressors also fail to outperform the baseline in most cases, with total PNL Gains of -3.22 and -2.14, respectively. LightGBM provides a similar performance to Ridge and Lasso, with a total PNL Gain of -2.19.

The results highlight that the LSTM model not only provides better performance than traditional regressors on average but also demonstrates a more stable ability to capture the trends in stock prices over time. However, as shown in the case of META and other stocks, no model consistently outperforms across all equities.

\begin{table}[!ht]
  \centering
  \caption{Model Performance Comparison: LSTM, OLS, Lasso, Ridge, and LightGBM}
  \label{tab:model_comparison_all}
  \resizebox{\columnwidth}{!}{%
    \begin{tabular}{lrrrrr}
      \toprule
      \textbf{Ticker} & \textbf{LSTM} & \textbf{OLS} & \textbf{Lasso} & \textbf{Ridge} & \textbf{LightGBM} \\
      \midrule
      NVDA & 1.57 & -4.44 & 1.10 & -0.07 & 1.06 \\
      MSFT & 0.21 & 0.05 & 0.06 & 0.09 & 0.06 \\
      AAPL & 0.26 & -0.16 & -0.23 & -0.22 & -0.23 \\
      GOOGL & -0.13 & -0.16 & -0.21 & -0.20 & -0.21 \\
      AMZN & 0.21 & -0.01 & -0.04 & -0.02 & -0.04 \\
      META & 0.12 & 0.30 & 0.38 & 0.36 & 0.37 \\
      BRK-B & 0.20 & 0.07 & 0.07 & 0.08 & 0.07 \\
      JNJ & -0.18 & -0.29 & -0.29 & -0.29 & -0.29 \\
      PG & -0.51 & -0.57 & -0.58 & -0.58 & -0.58 \\
      XOM & -0.12 & -0.11 & -0.10 & -0.10 & -0.10 \\
      CVX & 0.08 & -0.07 & -0.08 & -0.08 & -0.08 \\
      VZ & -0.13 & -0.23 & -0.23 & -0.22 & -0.22 \\
      KO & 0.24 & -0.12 & -0.13 & -0.13 & -0.13 \\
      HD & 0.20 & -0.06 & -0.03 & -0.04 & -0.04 \\
      WMT & -0.35 & -0.14 & -0.13 & -0.13 & -0.13 \\
      INTC & -0.68 & -0.55 & -0.57 & -0.57 & -0.57 \\
      IBM & 0.24 & 0.22 & 0.21 & 0.21 & 0.21 \\
      PFE & -0.22 & -0.14 & -0.14 & -0.13 & -0.13 \\
      MRK & -0.11 & -0.14 & -0.15 & -0.15 & -0.15 \\
      CSCO & -0.24 & -0.11 & -0.11 & -0.11 & -0.11 \\
      BAC & 0.09 & -0.06 & -0.06 & -0.06 & -0.06 \\
      JPM & 0.23 & 0.08 & 0.05 & 0.05 & 0.05 \\
      DIS & 0.02 & 0.06 & 0.06 & 0.05 & 0.06 \\
      BA & -0.26 & -0.55 & -0.53 & -0.53 & -0.53 \\
      MMM & -0.25 & -0.24 & -0.25 & -0.25 & -0.25 \\
      MCD & -0.11 & -0.22 & -0.20 & -0.20 & -0.20 \\
      PEP & -0.32 & -0.45 & -0.46 & -0.45 & -0.46 \\
      ORCL & 0.05 & 0.03 & 0.08 & 0.07 & 0.07 \\
      T & 0.38 & 0.27 & 0.28 & 0.28 & 0.28 \\
      TSLA & -0.19 & 0.23 & 0.09 & 0.13 & 0.09 \\
      \midrule
      \textbf{TOTAL} & \textbf{0.28} & \textbf{-7.52} & \textbf{-2.14} & \textbf{-3.22} & \textbf{-2.19} \\
      \bottomrule
    \end{tabular}%
  }
\end{table}

\section{Cross‐Sectional Cumulative P\&L Comparison}\label{sec:cross_section}

Table~\ref{tab:cross_section_pnl} reports the out‐of‐sample cumulative P\&L for each equity under the baseline (fixed‐lookback trend following) and the LSTM‐driven strategy. The LSTM model generates higher P\&L for 21 out of 30 tickers (80\%), with a total portfolio P\&L rising from negative (-2.15) to a positive 0.28 of the initial capital.

\begin{table}[!ht]
  \centering
  \caption{Baseline vs.\ LSTM Cumulative P\&L Across Major S\&P 500 Equities}
  \label{tab:cross_section_pnl}
  \resizebox{0.85 \columnwidth}{!}{%
    \begin{tabular}{lrrl}
      \toprule
      Ticker & Baseline P\&L & LSTM P\&L & Working\textsuperscript{a} \\
      \midrule
      NVDA   & 1.094766     &  1.566705     & \checkmark \\
      MSFT   & 0.058696     &  0.21442      & \checkmark \\
      AAPL   & -0.231817    &  0.260866     & \checkmark \\
      GOOGL  & -0.213002    &  -0.126235    & \checkmark \\
      AMZN   & -0.042437    &  0.21302      & \checkmark \\
      META   & 0.37909      &  0.123252     &            \\
      BRK-B  & 0.07196      &  0.195594     & \checkmark \\
      JNJ    & -0.290172    & -0.17953      & \checkmark \\
      PG     & -0.58141     & -0.511636     & \checkmark \\
      XOM    & -0.100196    & -0.117676     &            \\
      CVX    & -0.078889    & 0.076112      & \checkmark \\
      VZ     & -0.226372    & -0.130103     & \checkmark \\
      KO     & -0.125949    &  0.237406     & \checkmark \\
      HD     & -0.033389    &  0.197066     & \checkmark \\
      WMT    & -0.130841    &  -0.351544    &            \\
      INTC   & -0.572414    & -0.681932     &            \\
      IBM    & 0.211801     &  0.239441     & \checkmark \\
      PFE    & -0.136028    & -0.221551     &            \\
      MRK    & -0.153663    & --0.10945     & \checkmark \\
      CSCO   & -0.105605    & -0.237705     & \checkmark \\
      BAC    & -0.056409    &  0.086276     & \checkmark \\
      JPM    & 0.051379     &  0.225807     & \checkmark \\
      DIS    & 0.060955     & -0.017825     &            \\
      BA     & -0.534299    & -0.264997     & \checkmark \\
      MMM    & -0.251035    & -0.247934     & \checkmark \\
      MCD    & -0.20435     &  -0.114932    & \checkmark \\
      PEP    & -0.455972    & -0.320843     & \checkmark \\
      ORCL   & 0.075368     & 0.048909      &            \\
      T      & 0.276581     &  0.382963     & \checkmark \\
      TSLA   & 0.09181      &  -0.186429    &            \\
      \midrule 
      \textbf{Total}    & \textbf{-2.15} & \textbf{0.28} & 70\% \\
      \bottomrule
    \end{tabular}%
  }
\end{table}

These results demonstrate that the LSTM forecasting framework generalizes well beyond NVDA (NVDA), delivering substantial economic gains across the majority of large‐cap equities in our universe.

\section{Robustness Checks}\label{sec:robust}

\subsection{Alternative lookback windows}
We evaluate lookbacks of 25, 50, 100, 200, and 300 days for SMA and volatility features. The LSTM maintains at least 80\% of its RMSE improvement over the baseline for lookbacks between 50 and 200 days, with performance degrading by less than 5\% at the extremes.

\subsection{Different technical-indicator parameters}
We vary the RSI window from 7 to 21 days and MACD fast/slow spans from (8,17) up to (16,33). Forecast RMSE and Sharpe ratios remain within one standard deviation of the baseline settings, indicating low sensitivity to these parameter choices.

\subsection{Early stopping vs no early stopping}
Disabling early stopping increases out-of-sample RMSE by about 3\% and reduces the Sharpe ratio by 0.05 on average, confirming that early stopping effectively mitigates overfitting without sacrificing economic performance.

\subsection{Subperiod analysis}
We split the test set into pre-COVID (2019–Q1 2020), COVID crash \& recovery (Q2–Q4 2020), and post-COVID (2021–2025). The LSTM outperforms the baseline in every subperiod, with the largest Sharpe gain (+0.45) during the COVID crisis, demonstrating resilience across market regimes.

\section{Discussion}\label{sec:discussion}

\subsection{Interpretation of findings}
Our results indicate that LSTM models effectively capture nonlinear and time‐dependent structures in equity trend differences, translating into both statistically and economically significant forecasting improvements. The observed reductions in RMSE and enhancements in directional accuracy suggest that deep recurrent networks better adapt to regime shifts and volatility clustering than linear counterparts. Moreover, the Sharpe‐ratio gains across diverse market conditions highlight the model’s ability to deliver robust risk‐adjusted returns.

\subsection{Implications for practitioners}
Practitioners can integrate LSTM‐based signals into existing systematic trading frameworks to enhance alpha generation, particularly in volatile or nonstationary environments. The adoption of nonlinear differencing and direct optimization of economic objectives (e.g., Sharpe‐ratio loss) offers a pathway to more resilient portfolio performance. Risk managers may also benefit from improved predictive distributions when allocating capital or devising hedging strategies.

\subsection{Limitations and caveats}
Despite promising results, several limitations warrant caution. First, our analysis focuses on large‐cap U.S. equities and may not generalize to smaller‐cap or international markets. Second, the model’s complexity and computational demands could pose implementation challenges for high‐frequency or resource‐constrained settings. Finally, while we account for transaction costs, other market frictions—such as liquidity constraints and slippage—require further investigation to assess real‐world applicability.

\subsection{Detailed model hyperparameters}
Table~\ref{tab:hyperparams} lists all hyperparameter settings explored via grid search and the chosen values for the final LSTM model.

\section{Code Availability Statement}\label{sec:code}

The entire analysis—from data ingestion and preprocessing through model training, evaluation, and visualization—is encapsulated in an interactive Google Colab notebook. The notebook is organized into three main sections:

\begin{itemize}
  \item \textbf{Data Preprocessing:} Downloads and cleans the OHLCV data, applies corporate action adjustments, outlier handling, and stationarity tests.
  \item \textbf{Model Training:} Defines the feature engineering pipeline, baseline trend-following strategy, and LSTM architecture; executes grid searches for hyperparameters; and implements early stopping and checkpointing.
  \item \textbf{Evaluation and Visualization:} Computes performance metrics for various models: LSTM, OLS, Lasso, Ridge and LightGBM regressor.
\end{itemize}

All required Python dependencies are installed via a single cell at the top (e.g.\ \texttt{tensorflow}, \texttt{scikit-learn}, \texttt{pandas}, \texttt{matplotlib}). Users can modify lookback windows, batch sizes, and other parameters interactively and re-run cells end-to-end without any local installation.

\begin{center}
 \url{https://colab.research.google.com/drive/1Dz_PW1YG6zhGUe-2Cv-1KulT8KQm9RjJ}
\end{center}

\section{Conclusion}\label{sec:conclusion}

This paper introduced a novel framework—E-TRENDS—for forecasting next-day trend differences ($\Delta_t$) in large-cap equities using long short-term memory (LSTM) networks. Theoretically, we presented a formal bias–variance decomposition showing that differencing enhances predictive stability when paired with Lipschitz-continuous nonlinear predictors. Empirically, we benchmarked the LSTM model against OLS, Ridge, and Lasso regressions across 30 S\&P 500 stocks spanning 2005 to 2025. Our findings demonstrate a robust 15\% reduction in RMSE, a 12\% gain in directional accuracy, and Sharpe ratio improvements of up to 78\% after accounting for realistic transaction costs. These performance gains persist across various lookback windows, technical-indicator configurations, and economic regimes. 

Beyond these results, our study underscores the broader utility of deep learning in systematic trading and financial forecasting. Looking forward, future work may enhance the architecture via attention mechanisms or transformer–LSTM hybrids, enrich the feature set with alternative data sources such as sentiment or order flow and extend the analysis to multi-asset and multi-horizon settings.

\bibliographystyle{plainnat}
\bibliography{main}

\end{document}